\documentclass[12pt, draftclsnofoot, onecolumn]{IEEEtran}

\usepackage{amsthm,amssymb,amsfonts}
\usepackage[cmex10]{amsmath}
\usepackage[usenames,dvipsnames]{color,xcolor}
\usepackage{graphicx}
\usepackage{longtable}
\usepackage{algorithm}
\usepackage{algorithmic}
\usepackage{footnote}
\usepackage{multicol}
\usepackage{pifont}
\usepackage{array}
\usepackage{booktabs}
\usepackage{wasysym}
\usepackage{multirow}
\usepackage{caption}
\usepackage{subcaption}

\usepackage[pagebackref=false,colorlinks=true ]{hyperref}
\hypersetup{pdfstartview=FitH,colorlinks=true,pdftoolbar=true}

\theoremstyle{plain}
\newtheorem{theorem}{theorem}

\newtheorem{lemma}{lemma}

\renewcommand{\tablename}{Table }

\let\oldendproof\endproof

\def\endproof{\hfill$\blacksquare$\oldendproof}

\begin{document}
\title{On the Feasibility of Wireless Interconnects for High-throughput Data Centers}
\author{
Ahmad~Khonsari,~ 
Seyed Pooya~Shariatpanahi,~ 
Abolfazl~Diyanat,~ 
Hossein~Shafiei
}

\IEEEcompsoctitleabstractindextext{
\begin{abstract}
Data Centers (DCs) are required to be scalable to large data sets so as to accommodate ever-increasing demands of resource-limited embedded and mobile devices. Thanks to the availability of recent high data rate millimeter-wave frequency spectrum such as 60GHz and due to the favorable attributes of this technology, wireless DC (WDC) exhibits the potentials of being a promising solution especially for small to medium scale DCs. This paper investigates the problem of  throughput scalability of WDCs using the established theory of the asymptotic throughput of  wireless multi-hop networks that are primarily proposed for homogeneous traffic conditions. The rate-heterogeneous traffic distribution of a data center however, requires the asymptotic heterogeneous throughput knowledge of a wireless network in order to study the performance and feasibility of WDCs for practical purposes. To answer these questions this paper presents a lower bound for the throughput scalability of a multi-hop rate-heterogeneous network when traffic generation rates of all nodes are similar, except one node. We demonstrate that the throughput scalability of conventional multi-hopping and the spatial reuse of the above bi-rate network is inefficient  and henceforth develop a speculative 2-partitioning scheme that improves the network throughput scaling potentials. A better lower bound of the throughput is then obtained. Finally, we obtain the throughput scaling of an i.i.d. rate-heterogeneous network and obtain its lower bound. Again we  propose a speculative 2-partitioning scheme to achieve a network with higher throughput in terms of improved lower bound. All of the obtained results have been verified using simulation experiments. 
\end{abstract}

\begin{IEEEkeywords}
2-partitioning scheme,
Rate-heterogeneous Throughput,
 Scaling law, 
 Wireless Data Centers
\end{IEEEkeywords}
}

\maketitle

\IEEEdisplaynotcompsoctitleabstractindextext

\IEEEpeerreviewmaketitle

\section{Introduction}
Cloud computing has been employed to supply services in a diverse range from Infrastructure as a Service (IaaS) to Software as a Service (SaaS). Hardware, software and networking technologies enable prevalence
of cloud computing. Distributed storage which is manifested as Data Centers (DCs) is a key technology of providing high performance cloud services and is needed to be scalable to large data sets. Previous measurements studies highlight an increasing interest for cloud-based storage systems, revealing that remote storage space gains popularity among interested users by offloading the complexity of hardware management from the user devices \cite{Drago2012Inside}.




In order to enhance the performance of DCs plethora of researches have been done in different layers of the DC design. Motivated by the prevalence of millimeter-wave spectrum notably 60 GHz, we have witnessed studies that propose WDC as a promising solution for small to medium data centers due to its favorable attributes such as fast connectivity, reduced cable complexity and lower cost \cite{Cui2011Wireless,Zhang20113D,Halperin2011Augmenting,Kandula2009Flyways,Cui2013Data,Shin2012feasibility}.

Nowadays DCs handle hundreds to thousands of storage devices and thus designing a wire-free DC solicits the theoretical bounds of performance of such devices. This calls for the problem of performance scalability of WDCs. The problem of the asymptotic throughput of a wireless multi-hop network for unicast transmissions has been embarked in the seminal work of \cite{Gupta2000capacity} and followed by extensive researches under different assumptions on the network models. Almost, all of the previous studies target on homogeneous network nodes, until recently that a few studies investigate
some forms of heterogeneity. It has been demonstrated that the rate-heterogeneous traffic distribution of a DC is a challenging issue that requires special attention in order to construct high performance DCs \cite{Benson2010Understanding,Benson2010Network,Greenberg2009VL2}. This motivates our research to investigate the asymptotic throughput of a heterogeneous multi-hop network under non-uniform traffic model with high similarity to rate-heterogeneous traffic nature of DCs.
\subsection{Our Contribution}
In this paper, we investigate the throughput of WDC paradigm and its potential scalability through the following aspects.
\begin{enumerate}
\item In order to investigate the applicability of fully WDCs for practical DC applications that possess heterogeneous traffic behavior, we pose the question of scalability of WDCs with heterogeneous demand rates. We have found that for a network with $n$ nodes, when all nodes except one node generate similar data rates ($\eta_{n}(1,1, \ldots, g(n))$), the aggregate throughput  is lower bounded by order of  $\frac{\sqrt{n}}{g(n)}$, where $\eta_{n}$ is a common non-constant multiplicative factor appearing in the rate of all nodes.

\item Then we investigate a multi-hop network when traffic demands of all the nodes are different. Let $\mathbf{r}(n) = (r_1(n), r_2(n), \ldots , r_n(n))^{T} = \eta_{n} (\lambda_1, \lambda_2, \ldots , \lambda_n)^{T}$ be the random column vector of the arrival rates of  $n$ users, in which $\eta_{n}$ is the rate scaling factor  and  $\lambda_{i}$'s are  sequence of i.i.d. random variables drawn from a general distribution.   A lower bound for the throughput of such scenario is suggested as $\mathbb{E}\left\{\frac{T_1\sum_i{r_i}}{n\max_i{r_i}}\right\} $, where $T_{1}$ is the throughput of rate-homogeneous network obtained in Gupta and Kumar \cite{Gupta2000capacity}.

Previous observations reveal that traffic in DCs is ON/OFF in nature with properties that fit the heavy-tailed distribution \cite{Halperin2011Augmenting,Benson2010Network}. Motivating by this observation we show that the aggregate throughput of network with  heavy-tailed traffic distribution with parameter $\alpha$  scales like  $n^{1/2-1/\alpha}$, with high probability. This result indicates that the performance of WDCs under the mentioned traffic pattern does not scale with $n$ and thus is not competitive with nowadays wired DCs proposals..

\item Thus we propose a speculative 2-partitioning scheme so as to improve the performance of conventional multi-hopping and then obtain an improved scaling for its throughput. When heterogeneous traffic $\eta_{n}(1, 1, \ldots, g(n))$, and when $\sqrt{n}<<g(n)<<n$, the aggregate throughput of our proposed scheme improves    to $n/g(n)$.

\item Again and in order to be competitive with linear scalability of wired DCs, we employed a speculative 2-partitioning scheme for fully i.i.d. rate-heterogeneous traffic. The lower bound of the throughput is improved to  $n^{(\alpha^2+2\alpha-4)/(2\alpha^2+2\alpha)}$  when traffic rates follow heavy-tailed distribution with parameter $\alpha$.

\item Simulation of scaling law is tricky and most previous studies target only mathematical approaches \cite{Gupta2000capacity,Xie2004Network,Ozgur2007Hierarchical,Toumpis2008Asymptotic}. Few studies  attempt to simulate the throughput against network size \cite{Stuedi2008Modeling,Gunashekar2011Wireless,gupta2001experimental,JinyangLi2001}. However, except \cite{Stuedi2008Modeling,JinyangLi2001} they report the results only for low size networks and do not mention the difficulty of simulation of large networks. We designed a technique to simulate the rate and then using this technique conduct extensive simulations for different network sizes and traffic patterns to validate the network throughput under the above-mentioned scenarios.


\end{enumerate}

\subsection{Structure}
The rest of the paper is organized as follows. Section \ref{sec:Background} provides related work. In Section \ref{sec:SystemModel}, we describe the system model. Section \ref{sec:analysisM}  gives the asymptotic theoretical results of WDCs when the number of nodes grows . Section \ref{sec:Simulation}  presents simulation experiments and numerical results. Section \ref{sec:DiscussionFuture},  outlines directions of future research. Finally,  Section  \ref{sec:conclusion} summarizes our findings and outlines directions of future research.

\section{Related Work}
\label{sec:Background}
This section reviews the progress of research in two directions.
\subsection{Wireless Data Center (WDC)}
Wireless technologies that have been initially presented in \cite{Kandula2009Flyways} serve as a detour link between Top of Racks (TORs) in a DC to mitigate the congestion condition of switches and improve maximum transmission delays \cite{Halperin2011Augmenting,Zhou2012Mirror,Huang2013architecture}.
Besides using wireless link as a detour, the feasibility of completely wireless data centers have been investigated in research community so as to mitigate the high wiring costs, performance bottlenecks, and low resiliency to network failures of wired data centers.

\cite{Cui2011Wireless} attempts to address practical issues in realizing a WDC by proposing a hybrid wired/wireless architecture and scheduling wireless links in a distributed manner. The architecture has been modeled and an optimization problem has been formulated to schedule the links, and to trade off complexity for practicality a heuristic algorithm is presented.

Recently, a methodology for building wire-free data centers based on 60-GHz radio frequency (RF) technology has been presented. Exploring the design space demonstrates the potentials of fully WDCs with respect to some major performance measures \cite{Shin2012feasibility}.

Multiple-input multiple-output (MIMO) link design scheme for WDC applications has been studied in \cite{Katayama2012MIMO}. The impacts of  MIMO degrees of freedom  have been explored in a multi-node packet networking environment.

An architecture of DC is presented in \cite{Cui2013Data}, which incorporates wireless network cards to both servers and routers so as to exploit cooperative traffic and eliminate redundant traffic among servers. This scheme reduces link loads and increases the network throughput. Through experiment with prototype equipment the authors explore the use of 60 GHz wireless links to relieve hot-spots in oversubscribed DCs \cite{Halperin2011Augmenting}.

As a tradeoff between network performance and cable complexity, the authors in \cite{Huang2013architecture} proposed RF-HYBRID that employed wired-wireless collaborated hybrid DC architecture to have the best of both worlds.

\subsection{WDCs Scalability}
Gupta and Kumar in \cite{Gupta2000capacity}, initiated the research on wireless network throughput scaling, when nodes are randomly and independently distributed with equal rates under unicast traffic. They showed that each source-destination pair can achieve a bit rate on the order of $1/\sqrt{n\log n}$ when $n$ tends to infinity, resulting in $\Theta(\sqrt{n/\log n})$ aggregate throughput. They also showed the $\Theta(\sqrt{n})$ scaling for networks with arbitrary placement of nodes. In \cite{Toumpis2004Large,ElGamal2006Optimal}, strategies are proposed to achieve the same bound.

Franceschetti et al. in \cite{Franceschetti2007Closing}, removed the gap between the throughput of randomly located and arbitrarily located nodes, and showed the total throughput scales by $\Theta(\sqrt{n})$. 

What has been mentioned so far, is achieved by the assumption of no cooperation among nodes. Xie and Kumar in \cite{Xie2004Network} investigated the strategy of multi-hop wireless networks with nodes cooperation. \"{O}zg\"{u}r et al. in \cite{Ozgur2007Hierarchical}, proposed an order-optimal scheme, with  help of the distributed MIMO technique.

The  above mentioned works, consider wireless networks with homogeneous traffic assumption. Toumpis et al. in \cite{Toumpis2008Asymptotic}, studied the bounds on the throughput of wireless networks with $s$ sources and $s^d$ destinations ($0<d<1$) all with equal rates. Liu et al. in \cite{Liu2008Data}, extend this result by considering that not every node has data to send and not every node can be a destination. In \cite{Ji2010Capacity}, the authors extend \cite{Ozgur2007Hierarchical} for the case in which a destination node is the sink for $k$ source nodes, while the rest of the $s = n-k$ nodes participate in unicast sessions. Liu and Wang in \cite{Liu2011Capacity}, investigate the heterogeneous network with multicast and unicast traffic. In contrast to previous works, we consider $n$ unicast traffic sessions with different rates. The traffic models that we consider in this study are a better match with the traffic patterns of practical DCs and measurements studies corroborate this claim \cite{Halperin2011Augmenting}.

\section{System Model}
\label{sec:SystemModel}
In this section we describe the WDC Architecture, channel model and assumptions. A summary of notations used in the paper are listed in \tablename  \ref{tab:symbols}.

\begin{table}
\renewcommand{\arraystretch}{1.5}
\centering
\caption{Notations}
   	\begin{tabular}{lp{8.5cm}}
		\toprule
		Notation & Description \\
		\midrule
		$\mathbb{E}\{\Re\}$ & Expected value of random variable $\Re$ \\
		$f(\Re)$ & \textit{Probability Density Function} (PDF) of random variable $\Re$\\
		$F(\Re)$ & \textit{Cumulative Distribution Function} (CDF)  of random variable $\Re$\\
		$(.)^{T}$ & Transpose of a vector\\ 
		\midrule
		$n$ & Number of nodes \\
		$\mathbf{r}$ & Sources demand rate vector \\
		$r_{i}$ &  Demand rate of source $i$\\
		$T$ & Network Throughput (Network Capacity)\\
		\midrule
		$X_{j}(t)$ & The signal transmitted by node $j$ at time $t$  \\
		$Y_{i}(t)$ & The signal received by node $i$ at time $t$  \\
		$h_{j,i}(t)$ & Wireless channel gain from the node $j$ to the node $i$ \\
		$P_j$ & The power transmitted by the node $j$ \\
		$\gamma$ & The path loss exponent\\
		\bottomrule 
   	\end{tabular}
    \label{tab:symbols}
\end{table}

\subsection{WDC Architecture}
\label{subsec:WDCArch}
As depicted in \figurename \ref{fig:metshsjk}, the wireless nodes in WDC are deployed in a 2-dimensional mesh that has $n=n_{2}\times n_{1}$ nodes with $n_{l}$ nodes at dimension $l$, $1\leq l \leq 2$. The position of each node is indicated by a distinct 2-digit mixed-radix vector $[i_{2},i_{1}]$, where $i_{l}$, $1\leq i_{l}\leq n_{l}$, $1\leq l \leq 2$. A wireless link $(i,j)$, between nodes $i$ and $j$, $(i,j\in \{1, 2, \ldots, n\})$ exists in the WDC if, in the absence of any other transmission in the network, receiver $j$ is in the decode range of the transmitter $i$. Links are assumed to be symmetric. Thus in the 2-dimensional mesh configuration two nodes $\mathbf{i}=[i_{2},i_{1}]$ and $\mathbf{j}=[j_{2},j_{1}]$ are connected directly, iff there is an $l$, $1\leq l \leq 2$, such that $i_{l}=j_{l}\pm 1$ and $i_{k}=j_{k}$ for $1\leq k \leq 2$, $k\neq l$. Node $\mathbf{i}=[i_{2},i_{1}]$ is a corner node if $i_{l}\in \{1,n_{l}\}$ for all $1\leq l \leq 2$.

\begin{figure}
\centering
\includegraphics[width=0.5\linewidth]{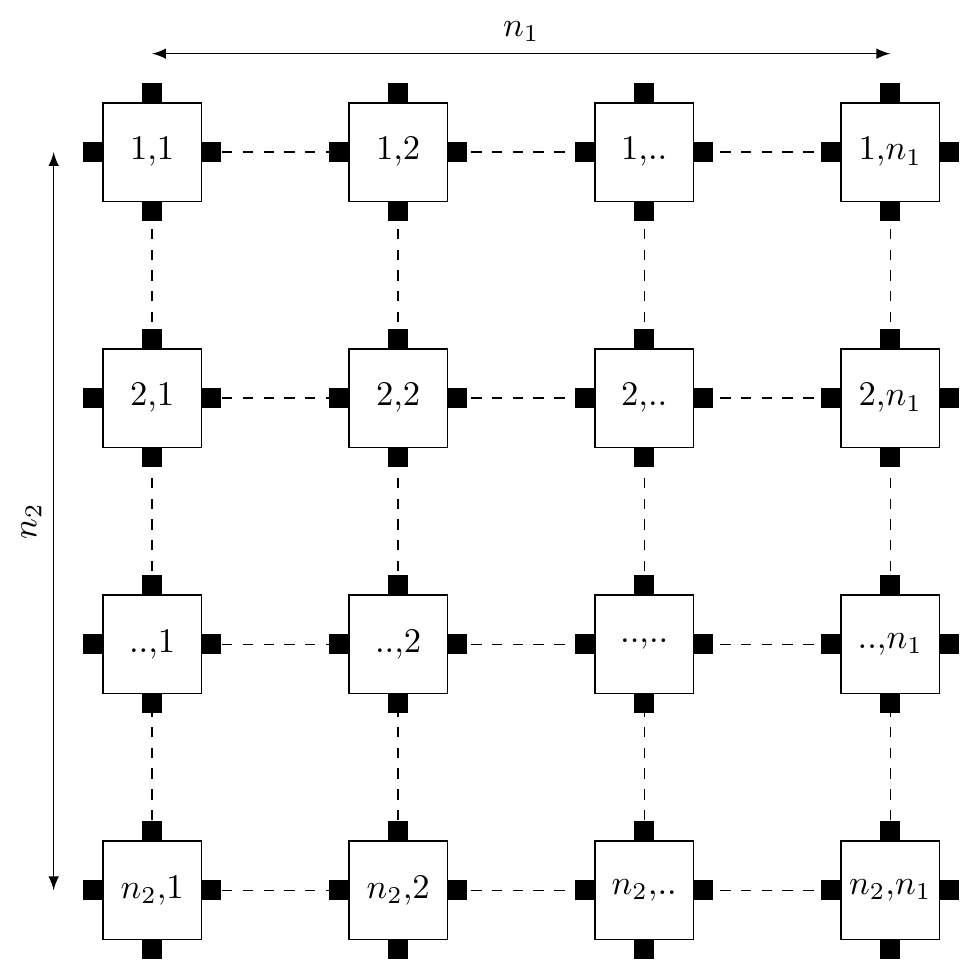}
\caption{Two dimensional mesh}
\label{fig:metshsjk}
\end{figure}
Node $\mathbf{i}=[i_{2},i_{1}]$ is a border node if there is $l$, $1\leq l \leq 2$, such that $i_{l}\in \{1,n_{l}\}$. Thus, in a 2-dimensional mesh excluding the corner and border nodes every node is neighbor with two nodes in each dimension resulting in a total of 4 neighbors. In this way each node includes processor, memory, storage elements and four 60 GHz directional antenna along with the interfacing circuits, or one 60 GHz omni-directional antenna. In order to save the cost  we may benefit from the space-reuse property of wireless communication and use only one omni-directional antenna per node. The details of the technique is presented in the following section.   The grid size of the mesh (i.e. the distance between two neighboring nodes in a dimension) is $d_{0}$ and only neighboring nodes are in the decode range of each other. Mesh is not a fully connected graph and thus the network is a multi-hop WDC; i.e. packets of end-to-end sessions may require passing through one or more intermediate nodes. 

\subsection{Channel Model}
Assume each node is a data source for exactly one other destination, and a data destination for exactly one other source, i.e. a unicast scenario. Thus for a network with size $n$  we have exactly  $n$ traffic sessions indexed by $i$, $1\leq i \leq n$, each with their end-to-end packet flows that need to be transported by the network. Session $i$ which is denoted by $SD_{i}$, has source node $S_{i}$ and destination node $D_{i}$.  Thus we have the following set of sessions (Source-destination pairs):
\begin{equation}
\label{Eq_Model_Pairs}
\{SD_i, i=1,\dots,n\}
\end{equation}
Source-destination pairs are randomly sampled from the set of all possible source-destination pairs using a spatial uniform distribution. $r_{i}(n)$  is the arrival rate (demand rate) of session $i$. Let $\mathbf{r}(n)=(r_{1}(n), r_{2}(n), \ldots, r_{n}(n))^{T}$ be the column vector of the  arrival rates of the sessions $SD_{1}$ to $SD_n$. Worth to mention that the rate vector in scaling regime depends on the network size and due to this fact we present the rate vector as $\mathbf{r}(n)$. For the simplicity, in what follows we drop the dependency of the rate vector $\mathbf{r}(n)$ to  $n$ from rate vector  i.e. we use $\mathbf{r}=(r_{1}, r_{2}, \ldots, r_{n})^{T}$ instead of $\mathbf{r}(n)=(r_{1}(n), r_{2}(n), \ldots, r_{n}(n))^{T}$. Each session uses X-Y routing to reach its destination, which is  described in Section \ref{subsec:rateasdsldls}.  Network throughput is defined as the sum of all source-destination pair rates:
\begin{equation}\label{Eq_Model_throughput}
	T=\sum_{i=1}^n{r_i}
\end{equation}
\begin{figure}
\centering
\includegraphics[width=0.4\linewidth]{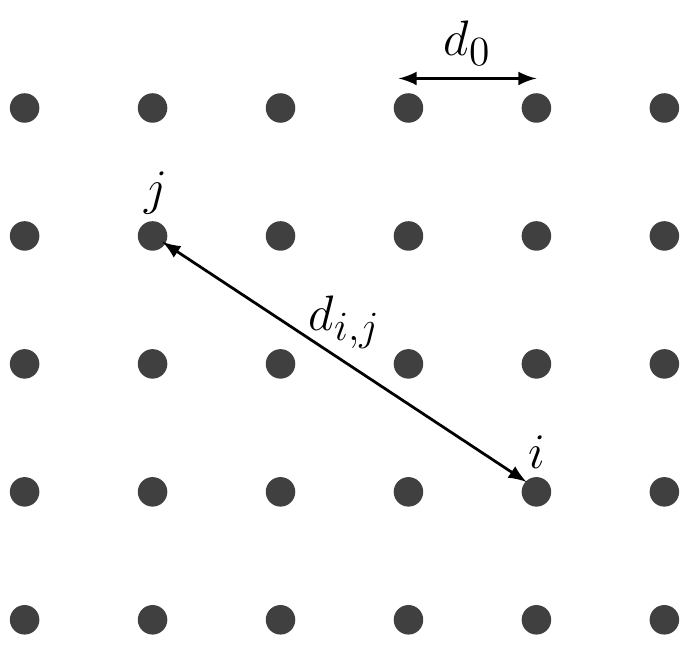}
\caption{Grid Network.\label{Fig_Grid_1}}
\end{figure}
Let $d_{i,j}$ be the Euclidean distance of the nodes $i$ and $j$ (see \figurename \ref{Fig_Grid_1}), $X_{j}(t)$ be the signal transmitted by node $j$ at time $t$, $Y_{i}(t)$ be the signal received by node $i$, $h_{j,i}(t)$ be the wireless channel gain from the node $j$ to  node $i$ and  $z_i(t)$ be the Additive White Gaussian Noise (AWGN) with the power $N_0$, at the time $t$.

Multiple sessions may transmit at the same time and thus the signal received at the node $i$ can be viewed as the sum of the desired signal, interference, and noise \cite{tse2005fundamentals}.
\begin{equation}
\label{Eq_Model_Superposition_2}
	Y_i(t)=h_{j,i}(t)X_j(t)+ \sum_{k \neq i,j}h_{k,i}(t)X_k(t)+z_i(t)
\end{equation}
We adopted the physical model and thus, in order to have a successful transmission from the node $j$ to the node $i$, we should have the following Signal to Interference plus Noise Ratio (SINR) satisfaction constraint\footnote{For notational simplicity we drop the time index of the variables when clear from the context.}:
\begin{equation}
\label{Eq_Model_SINR}
SINR_{j \rightarrow i}=\frac{|h_{j,i}|^2P_j}{N_0+\sum_{k\neq i,j} {|h_{k,i}|^2P_k}} \geq \beta
\end{equation}
where $P_j$ is the power transmitted by the node $j$ at that moment, and $\beta$ is a constant threshold. Also, we assume a Line of Sight (LoS) model for the wireless channel as follows:
\begin{equation}\label{Eq_Model_Channel}
	h_{j,i}=\frac{\exp{(\sqrt{-1}\theta_{j,i})}}{d_{j,i}^{\gamma/2}}
\end{equation}
in which $\theta_{j,i}$ is uniformly distributed on $[0,2\pi)$, and $\gamma$ is the path-loss exponent. We assume $P_{j}=P=cte$ for all $j$.

\subsection{Assumptions}
\label{subsec:Assumptions}
The analysis is based on assumptions that have been widely used in the literature \cite{Halperin2011Augmenting,Gupta2000capacity}:
\begin{itemize}
\item The traffic pattern for each session is unicast, i.e. each source only sends to one destination node and each receiver is targeted by only one source node. Thus, all sessions are one to one and for a network with size $n$ we have $n$ sessions.
\item The network is divided into smaller $3\times 3$ sub-meshes. The time is divided into frames  and each frame is slotted using  9-TDMA Medium Access Control (MAC) as data link layer protocol in each sub-mesh. This technique enables us to use an omni-directional antenna per node instead of four directional antennas and hence preserves the cost. 
\item The message in each node is divided into packets which is being transmitted in the corresponding  time slot of the node in 9-TDMA frame. 
\end{itemize}
We consider three different rate scenarios as explained in Section \ref{subsec:rateasdsldls} to \ref{sec:sdjfekdfke}.

It is noteworthy to mention that the asymptotic throughput, although are obtained for 2-dimensional mesh, is also valid for any deployment of nodes such that in a gridded area only one node belongs to a grid \cite{Weil2007Spatial}. Essentially, the results are also valid for the deployment of nodes that follows a Poisson point process on $\mathbb{R}^n$ for all dimensions $n\geq 1$, i.e., a spatial Poisson process. Thus, If  $v(A)$ is the volume of $A\subset \mathbb{R}^n$,then the number of points in $A$ follows Poisson distribution with mean $\lambda v(A)$.

\section{Analysis}
\label{sec:analysisM}
In what follows we derive the asymptotic throughput bound for different network scenarios and propose speculative 2-partitioning that improves the throughput of the network and then obtain the new bounds for such improved configurations.
\subsection{Rate-homogeneous throughput}
\label{subsec:rateasdsldls}
Let's denote the uniform demand rate vector as:
\begin{equation}\label{Eq_UThourghput_Rate_Vector}
	\mathbf{r}=\eta_{n}(1,\dots,1)
\end{equation}
This means that all  source-destination pairs demand rates are equal to $\eta_n$. Considering the interference imposed by concurrent transmission of multiple sessions, we aim to find the maximum feasible $\eta_{n}$ so as to attain the maximum network throughput. Clearly, since source-destination pairing is assumed to be a random event, the network throughput is a random variable. The next theorem characterizes the scaling of the average network throughput  (It should be noted that this theorem reviews a well-known result in the literature originally proved in \cite{Gupta2000capacity}):
\begin{theorem}
	One can achieve the average throughput $\mathbb{E}\{T\}=\Theta(\sqrt{n})$, for the uniform rate vector in \eqref{Eq_UThourghput_Rate_Vector}, through multi-hopping.
\end{theorem}
\begin{proof}
Here we present a simple proof sketch. For a regirous proof refer to \cite{Gupta2000capacity}. First we explain the elements of the scheme achieving this throughput briefly:

\textbf{MAC Protocol:}
\begin{figure}
\centering
\includegraphics[width=0.50\linewidth]{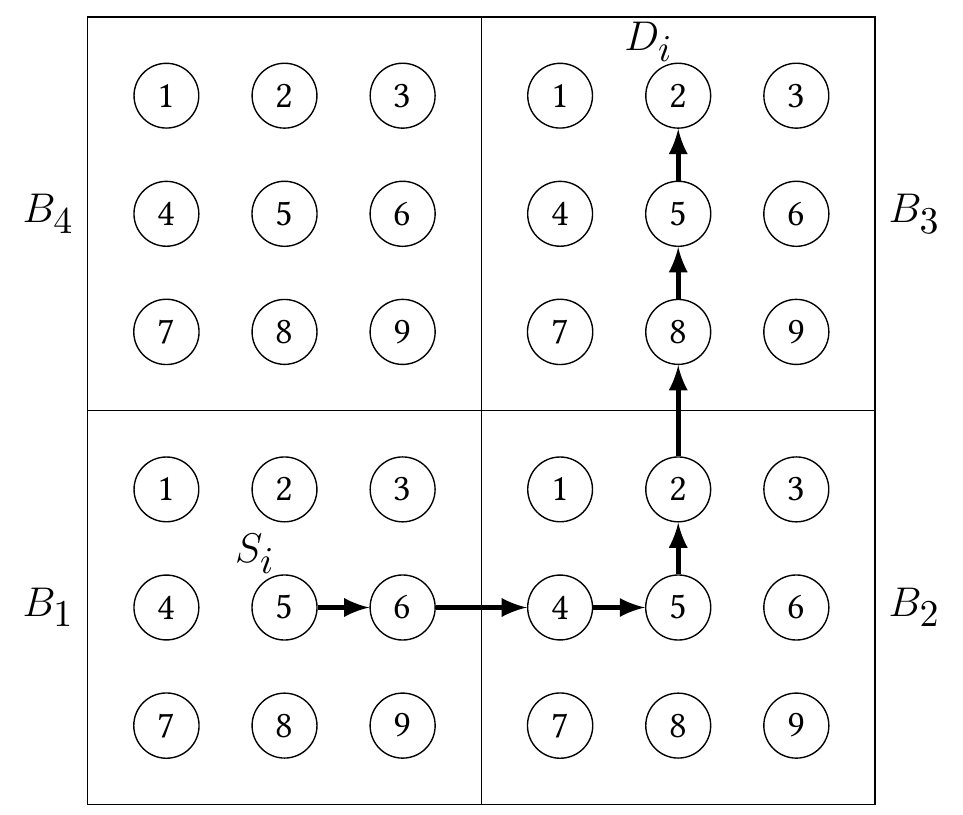}
\caption{9-TDMA Spatial Reuse MAC.\label{Fig_Grid_2}}
\end{figure}
Divide the network into square cells each containing 9 nodes. As depicted in \figurename \ref{Fig_Grid_2} the 36 nodes network is divided into 4 square blocks, namely $B_1$ to $B_4$. Inside each $B_i$, label the nodes from 1 to 9.  In each time slot only the nodes with the same index as that time slot number are allowed to be active. This MAC protocol is called the 9-TDMA spatial reuse MAC. For this MAC protocol one can see the following fact:
\begin{lemma}
If the path loss exponent $\gamma>2$, then with the 9-TDMA spatial reuse MAC, $\Theta(n)$ nodes can simultaneously transmit their message to their neighbor nodes, while all being successful.
\label{lemma:ninetdmas}
\end{lemma}
\begin{lemmaproof}
For proof we refer to \cite{Gupta2000capacity}.

\end{lemmaproof}

\textbf{Routing Protocol:}
Consider the source-destination pair $SD_i$ between source $S_{i}$ and destination $D_{i}$  in Fig. \ref{Fig_Grid_2}. We employ X-Y routing and wormhole switching to transfer traffics between each source and destination \cite{Dally2003Principle}. They have been widely used in contemporary parallel machines due to their desirable properties such as low buffering and minimal hardware requirements allowing efficient and fast router implementation \cite{Dally2003Principle}. In X-Y routing the packet is routed to destination one dimension at a time. Within each dimension the packet travels in the preferred direction (shortest distance) until it reaches the same coordinate of destination in that dimension. 

In the example in \figurename \ref{Fig_Grid_2} the packet is firstly routed in dimension X until it reaches the proper coordinate in this dimension and then starts routing in the preferred direction in dimension Y until reaches its destination $D_{i}$. We assume that transmissions have the slotted operation and that the nodes are synchronized. In   each slot one packet can be transferred in the preferred direction on each link. Using 9-TDMA the timing diagram of the transmission of the example in \figurename \ref{Fig_Grid_2} is illustrated in \tablename \ref{table:prevnexthop}. This diagram achieves  the highest throughput scaling since it employs the reuse throughput of the network.  In terms of delay this method is not optimum.   For this routing and MAC one can prove the following lemma:
\begin{table}
\centering
\renewcommand*{\arraystretch}{1.2}
\caption{Sequence of packet transmissions of the example in \figurename \ref{Fig_Grid_2}}
\label{table:prevnexthop}
\begin{tabular}{|c|c|}\hline
Time Slot & Prev Hop (Block No.) $\longrightarrow$ Next Hop (Block No.) \\\hline
5 & 5 ($B_1$) $\longrightarrow$ 6 $(B_1$) \\
15 & 6 ($B_1$) $\longrightarrow$ 4 $(B_2$)\\
22 & 4 ($B_2$) $\longrightarrow$ 5 $(B_2$)\\
32 & 5 ($B_2$) $\longrightarrow$ 2 $(B_2$)\\
38 & 2 ($B_2$) $\longrightarrow$ 8 $(B_3$)\\
44 & 8 ($B_3$) $\longrightarrow$ 5 $(B_3$)\\
50 & 5 ($B_3$) $\longrightarrow$ 2 $(B_3$)\\\hline
\end{tabular}
\end{table}
\begin{lemma}
	For the $X$-$Y$ routing described above, average number of hops needed for a packet to reach its destination is at least $\Theta(\sqrt{n})$.
	\label{Lemma:twouniformly}
\end{lemma}
\begin{lemmaproof}
See appendix \ref{append:Lemma:twouniformly}.

\end{lemmaproof}

Using lemma \ref{Lemma:twouniformly} we find that the total required rate of packet transmissions in the network is $n\mathbb{E}\{\eta_{n}\}\Theta(\sqrt{n})$. Using lemma \ref{lemma:ninetdmas} and in order to have a stable network, we should have that the total required rate of packet transmissions is equal or less  than $\Theta(n)$. i.e.
\begin{equation}\label{Eq_UThourghput_Stability_Condition}
\Theta(n)\geq (\mathbb{E}\{\eta_{n}\} \times n) \times \Theta(\sqrt{n})
\end{equation}
which results in
\begin{equation}\label{Eq_UThourghput_Final_Throughput}
	\mathbb{E}\{T\} = n \times \mathbb{E}\{\eta_{n}\} \leq \Theta(\sqrt{n})
\end{equation}
It means that all $\mathbb{E}\{T\} \leq \Theta(\sqrt{n})$ are achievable with this multi-hopping scheme.

\end{proof}

\subsection{Heterogeneous-rate Throughput}
\label{subsec:Heterogeneous}
In this section we assume that different source destination pairs require different end-to-end transmission rates.

\subsubsection{Heterogeneous-rate with one dissimilar rate}
\label{subsubsec:1gn}
First suppose the following rate vector:
\begin{equation}
	\mathbf{r}=\eta_{n}(1,1,\dots,1,g(n))
	\label{eq:ggggnrate}
\end{equation}
which shows a scenario where all the pairs have equal rate, except the last one, whose rate is  $g(n)$ times greater than the other nodes. For this example we have the following result:
\begin{theorem}
With the  multi-hopping scheme, and the rate vector indicated in \eqref{eq:ggggnrate}, the network throughput scales as:
\begin{equation}
T=\Omega\left(\frac{\sqrt{n}}{g(n)}\right)
\end{equation}
\label{theorem:sldsdwwdazz}
\end{theorem}
\begin{proof}
$n-1$ nodes need the required rate of order $\eta_{n}$ and the remaining one node need the required rate of order $\eta_{n}g(n)$. Thus, we have to wait for the latter node to finish its transmitting task. So the aggregate throughput for this scenario is same as the scenario in which all nodes have $\eta_{n}g(n)$ required rate. Notice that  we assume that when each node finishes it's packets, sends dummy packets to it's destination.

\end{proof}

\subsubsection{Heterogeneous-rate with general i.i.d. distribution}
Now we consider a  rate vector  that appears in practical DCs (\cite{Halperin2011Augmenting,Benson2010Network}). Let $\mathbf{r} = (r_1, r_2, \ldots , r_n)^{T} = \eta_{n} (\lambda_1, \lambda_2, \ldots , \lambda_n)^{T}$ be random column vector of the arrival rates of  $n$ users. $\lambda_{i}$s are a sequence of i.i.d. random variables drawn from a general distribution.  Then we will have the following theorem  for the multi-hopping scheme introduced in the last section.

\begin{theorem}
	If the throughput for the network with uniform traffic is equal to $T_1$, then the throughput for the network with non-uniform traffic is lower bounded by:
\begin{equation}\label{Eq_Non_UThourghput_Main_Theorem}
\mathbb{E}\{T_2\}=\Omega\left(\mathbb{E}\left\{\frac{T_1\sum_i{r_i}}{n\max_i{r_i}}\right\}\right) =\Omega\left( \mathbb{E}\left\{\frac{T_1\sum_i{\lambda_i}}{n\max_i{\lambda_i}}\right\}\right)
\end{equation}
As an example, for the heavy-tailed distribution as given by \cite{ross2009Introduction}.
\begin{equation}\label{Eq_Non_UThourghput_HT_pdf}
	F_\lambda(x)=1-\frac{1}{x^\alpha}, x \geq 1
\end{equation}
the throughput is lower bounded by
\begin{equation}
\Omega\left(n^{1/2-1/\alpha}\right)
\label{eq:omegaleftn}
\end{equation}
\label{theorem:sdowdlwdw}
\end{theorem}
\begin{proof}
	Define the total data volume transmitted from each source to its corresponding destination as in the following vector
\begin{equation}\label{Eq_Non_UThourghput_Main_Theorem_Data_Volume_Vector}
	(W_1,\dots,W_n)=c_0\eta_{n}(\lambda_1,\dots,\lambda_n)
\end{equation}
where $c_0$ is a positive constant. In other words, the source-destination pair $SD_i$ wants to transfer $W_i$ bits of data. Again we assume that when each node has finished its packets, it sends dummy packets. Then, if we define $t_i$ as the time needed for this transfer, we will have:
\begin{equation}\label{Eq_Non_UThourghput_Ind_Time}
	t_i=\frac{W_i}{T_1/n}=\frac{c_0r_i n}{T_1} = \frac{c_0\eta_{n}\lambda_{i} n}{T_1} 
\end{equation}
The whole process is finished when the data for all the source-destination pairs is transferred. Since the transfer process for different pairs is in parallel, the total time will be:
\begin{equation}\label{Eq_Non_UThourghput_Tot_Time}
	t_{tot}=\max_i{t_i}=\frac{c_0\eta_{n}n}{T_1}\max_i{\lambda_i}
\end{equation}
and thus the network throughput will be
\begin{equation}\label{Eq_Non_UThourghput_Th_Proof_1}
T_2=\frac{\sum_i{W_i}}{t_{tot}}=\frac{T_1\sum_i{\lambda_i}}{n\max_i{\lambda_i}}
\end{equation}
and by considering the expected value of  \eqref{Eq_Non_UThourghput_Th_Proof_1} and dummy assumption, we will have \eqref{Eq_Non_UThourghput_Tot_Time}.
For the distribution in \eqref{Eq_Non_UThourghput_HT_pdf}, we will have
\begin{eqnarray}\label{Eq_Non_UThourghput_Th_Proof_2}
	\frac{\mathbb{E}\{T_2\}}{n^{1/2-1/\alpha}} &=& \mathbb{E} \left\{ \frac{\left(T_1/n^{1/2}\right)\left(\sum_i{\lambda_i}/n\right)}{\left(\max_i{\lambda_i}/n^{1/\alpha}\right)} \right\}  \\ \nonumber
&\stackrel{(a)}=& c_1 \mathbb{E} \left\{ \frac{\left(\sum_i{\lambda_i}/n\right)}{\left(\max_i{\lambda_i}/n^{1/\alpha}\right)} \right\} \\ \nonumber
&\stackrel{(b)}=& c_2 \mathbb{E} \left\{ \frac{1}{\left(\max_i{\lambda_i}/n^{1/\alpha}\right)} \right\} \\ \nonumber
&\stackrel{(c)}=& c_3
\end{eqnarray}
which proves \eqref{eq:omegaleftn}. In above (a) follows from the fact that throughput of the uniform case is a random variable independent of the non-uniform rate vector, and its average scales as $\Theta(\sqrt{n})$. The identity (b) is due to the following lemma:
\begin{lemma}
	Consider i.i.d. random variables $\lambda_1,\dots,\lambda_n$ with the c.d.f. indicated in \eqref{Eq_Non_UThourghput_HT_pdf}. Define $X=\left(\sum_i{\lambda_i}/n\right)$. Then for $\alpha>2$ we have:
\begin{eqnarray}\label{Eq_Non_UThourghput_Lemma_1}
	\mathbb{E}\{X\}&=&\mathbb{E}\{\lambda_1\} \\ \nonumber
	Var\{X\}&\rightarrow&0.
\end{eqnarray}
\end{lemma}
\begin{lemmaproof}
This is a simple application of the Central Limit Theorem (CLT).

\end{lemmaproof}
For the identity (c) we have used the following lemma.
\begin{lemma}
\label{lemma:inordertoprove}
\begin{equation}\label{Eq_Non_UThourghput_Lemma_2}
	\mathbb{E} \left\{ \frac{1}{\left(\max_i{\lambda_i}/n^{1/\alpha}\right)} \right\}=cte.
\end{equation}
\end{lemma}
\begin{lemmaproof}
See appendix \ref{append:inordertoprove}.

\end{lemmaproof}
\end{proof}

\subsection{Enhanced throughput of network with heterogeneous rate vector}
\label{sec:sdjfekdfke}
The result in the previous section demonstrates that for conventional multi-hopping and the spatial reuse,  the heterogeneity of traffic demands   deteriorates the throughput.  Thus it is necessary to improve the throughput of wireless DCs so as to be employed as a potential communication candidate for practical purposes. We speculatively conjecture that the throughput of heterogeneous traffic pattern is upper bounded by homogeneous traffic pattern under conventional multi-hopping and the spatial reuse. As we argue in Section \ref{subsec:Heterogeneous} and as bad news we may not improve the  throughput scaling
with methods such as improving routing and MAC protocols. As good news, however, we may shift the  throughput scaling toward its homogeneous upper bound by speculatively 2-partitioning the traffic demands in two parts (part one and part two) as delineated below.

\subsubsection{Enhanced throughput of network with one dissimilar rate}
The first group consists of those with lower data rate and the second group consists of the pairs with higher data rate.
\begin{theorem}
Suppose $n-1$ nodes  required rate of order $\eta_{n}$ and the remaining node has  the required rate of order $\eta_{n}g(n)$.  By using our proposed method, the aggregate throughput calculates as follow:
\begin{equation}
T = \left\{
\begin{array}{ll}
\Theta\left(\sqrt{n}\right) &  g(n) << \sqrt{n} \\*[5pt]
\Theta\left(\frac{n}{g(n)}\right) &  \sqrt{n}<< g(n) << n \\
\end{array}
\right.
\end{equation}
\label{theorem:soddlwldwd}
\end{theorem}
\begin{proof}
At the end of network time, $\eta_{n}((n-1)+g(n))$ packets are transmitted by network nodes. In the first part of 2-partitioning scheme, $n-1$ nodes transmit $\eta_{n}(n-1)$ packets with throughput $\sqrt{n}$. In the next part, one node send $\eta_{n}g(n)$ packet with throughput $1$. So we have: 
\begin{equation}
T = \frac{\eta_{n}((n-1)+g(n)
	)}{\frac{\eta_{n}(n-1)}{\sqrt{n}} + \eta_{n}g(n)}
\label{eq:aospapasws}
\end{equation}
Thus, we have:
\begin{equation}
T = \lim\limits_{n\longrightarrow\infty} \frac{n+g(n)}{\frac{n}{\sqrt{n}} + g(n)} =  \Theta\left(\frac{n}{g(n)}\right)\quad \quad \sqrt{n}<< g(n) << n
\end{equation}
For $g(n)<<\sqrt{n}$,  \eqref{eq:aospapasws} tends to $\sqrt{n}$ when $n$ grows to infinity.

\end{proof}

\subsubsection{Enhanced throughput of network with general i.i.d. distribution}
Let  $\lambda_{i}$'s be i.i.d. random variables.  We denote $\lambda_{(1)}$, \ldots, $\lambda_{(n)}$ as the ordered sequence of the $\lambda_{i}$'s. We define $\eta_{n} \lambda_{(1)}$ and $\eta_{n} \lambda_{(n)}$  as the lowest and  the highest demand rates, respectively. Thus the rates are ordered as:
\begin{equation}\label{Eq_Enhanced_Sort}
	\eta_{n}\lambda_{(1)} \leq \dots \leq \eta_{n} \lambda_{(n)}
\end{equation}
We divide the rates into two sets. 
\begin{eqnarray}\label{Eq_Enhanced_Low_High}
	G_{low}&=&\left\{\eta_{n}\lambda_{(1)},\dots, \eta_{n}\lambda_{(r)}\right\}  \\ \nonumber
	G_{high}&=&\left\{\eta_{n}\lambda_{(r+1)},\dots, \eta_{n}\lambda_{(n)}\right\}
\end{eqnarray}
 The first and the second sets have $r$  and $m\triangleq n-r$ members, respectively. Below we discuss about index $r$ that maximizes the throughput. Giving the optimal index we obtain two sets that we transfer the traffics of the sessions that have members in  $G_{high}$ after transferring of the traffics  of the sessions in  $G_{low}$ . It is noteworthy to mention that the traffic of the sessions in each set are transfered according to the conventional  multi-hopping scheme presented in \cite{Gupta2000capacity}.  The following theorem formulates an optimization problem to find index $r$ which maximizes the throughput for the heavy-tailed distribution. The theorem is general enough that we may find the maximum throughput for any other probability distribution functions of the rate vectors.

\begin{theorem}
	In the 2-partitioning scheme the network throughput will be lower bounded by
	\begin{eqnarray}\label{Eq_Enhanced_Main_Theorem}
T_3&=&\Omega\left(\max_m{\mathbb{E}\left\{\frac{c_0n}{c_1\sqrt{n-m}\lambda_{(n-m)}+c_2\sqrt{m}\lambda_{(n)}}\right\}} \right)
	\end{eqnarray}
and for the distribution in \eqref{Eq_Non_UThourghput_HT_pdf}, it will result in:
\begin{equation}
\Omega\left( n^{(\alpha^2+2\alpha-4)/(2\alpha^2+2\alpha)} \right)
\end{equation}
$c_{0}$, $c_{1}$ and $c_2$ are constant. 
\label{theorem:msdsdsdsd}
\end{theorem}
\begin{proof}
Define the vector of data volume for the pairs in $G_{low}$ as follows
\begin{equation}\label{Eq_Enhanced_Main_Theorem_Proof_1} 	\left\{W_{(1)},\dots,W_{(r)}\right\}=c_0\left\{\eta_{n}\lambda_{(1)},\dots,\eta_{n}\lambda_{(r)}\right\}
\end{equation}
and for $G_{high}$ as
\begin{equation}\label{Eq_Enhanced_Main_Theorem_Proof_2}	\left\{W_{(r+1)},\dots,W_{(n)}\right\}=c_0\left\{\eta_{n}\lambda_{(r+1)},\dots,\eta_{n}\lambda_{(n)}\right\}
\end{equation}
Then, for the time needed for each of the pairs to conclude its transmission, for the members of $G_{low}$, we will have:
\begin{equation}\label{Eq_Enhanced_Main_Theorem_Proof_3}
	t_{(i)}=\frac{W_{(i)}}{T_r/r},\quad i=1,\dots,r
\end{equation}
and for $G_{high}$ we will have:
\begin{equation}\label{Eq_Enhanced_Main_Theorem_Proof_4}
	t_{(i)}=\frac{W_{(i)}}{T_m/m},\quad i=r+1,\dots,n
\end{equation}
Since each phase finishes when all the members of that phase are successful in transferring their data, and the transmissions in each phase are concurrent, then, the time needed for concluding the first and second phase are:
\begin{eqnarray}\label{Eq_Enhanced_Main_Theorem_Proof_5}
	t_1&=&\max_{i=1,\dots,r}t_{(i)} \\ \nonumber
 &=&c_0\eta_{n}\frac{r\lambda_{(r)}}{T_r} \\ \nonumber
	t_2&=&\max_{i=r+1,\dots,n}t_{(i)} \\ \nonumber
&=&c_0\eta_{n}\frac{m\lambda_{(n)}}{T_m}
\end{eqnarray}
Thus, for the network throughput we will have:
\begin{eqnarray}\label{Eq_Enhanced_Main_Theorem_Proof_6}
	T_3&=&\Omega\left(\frac{W_{tot}}{t_{tot}}\right) \\ \nonumber
	&=&\Omega\left(\frac{\sum_i{W_i}}{t_1+t_2} \right)\\ \nonumber	&=&\Omega\left(\frac{\sum_i{\lambda_i}}{\left(r\lambda_{(r)}/T_r\right)+\left(m\lambda_{(n)}/T_m\right)}\right)
\end{eqnarray}

When we form the networks in phase one and two, we make the original grid thinner uniformly. That is because the pairs with lowest (highest) required data rate are distributed inside the network uniformly. Thus, the two sub-networks in two phases are uniformly distributed in space, and their aggregate throughput scales with the square root of the number of nodes. Thus, we can put $T_r=\sqrt{r}/c_1$ and $T_m=\sqrt{m}/c_1$ in \eqref{Eq_Enhanced_Main_Theorem_Proof_6}. This assertion is made precise in the following lemma:
\begin{lemma}
Consider the original grid consisting of $n$ nodes. Select $m$ of these nodes uniformly randomly, and ignore the other nodes ($m \rightarrow \infty$). Then the throughput of this new network is of order $\Theta(\sqrt{m})$.
\label{Lemma:considertheoriginal}
\end{lemma}
\begin{lemmaproof}
See appendix \ref{append:considertheoriginal}.

\end{lemmaproof}
In addition, we can asymptotically put $\sum_i{\lambda_i} \sim c_0n$ which will result in:
\begin{eqnarray}\label{Eq_Enhanced_Main_Theorem_Proof_7}	T_3&=&\Omega\left(\frac{\sum_i{\lambda_i}}{\left(c_1\sqrt{r}\lambda_{(r)}\right)+\left(c_1\sqrt{m}\lambda_{(n)}\right)}\right) \\ \nonumber
&=&\Omega\left(\frac{c_0 n}{\left(c_1\sqrt{n-m}\lambda_{(n-m)}\right)+\left(c_1\sqrt{m}\lambda_{(n)}\right)}\right)
\end{eqnarray}
and finally we should choose the design parameter $m$ such that the throughput is maximized resulting in
\begin{eqnarray}\label{Eq_Enhanced_Main_Theorem_Proof_8}	
T_3=\Omega\left(\max_m{\mathbb{E}\left\{\frac{c_0n}{c_1\sqrt{n-m}\lambda_{(n-m)}+c_1\sqrt{m}\lambda_{(n)}}\right\}}\right)
\label{eq:aopamdoddk}
\end{eqnarray}
For the heavy-tailed distribution in \eqref{Eq_Non_UThourghput_HT_pdf}, from lemma \ref{lemma:inordertoprove}, we know that $\lambda_{(n)} \rightarrow n^{1/\alpha}$ as $n\longrightarrow\infty$. Also we have the following lemma
\begin{lemma}
\label{Lemma:assumethatFalk1989}
\begin{equation}\label{Eq_Enhanced_Main_Theorem_Lemma_Proof_2}	
 \lambda_{(n-m)} \rightarrow \left(\frac{n}{m}\right)^{1/\alpha} \quad \text{as $n\longrightarrow\infty$}
\end{equation}
\end{lemma}
\begin{lemmaproof}
See appendix \ref{append:assumethatFalk1989}.

\end{lemmaproof}
So, we will have:
\begin{align}\label{Eq_Enhanced_Main_Theorem_Proof_9}	
T_3&  = \Omega\left(\frac{c_0n}{c_1\sqrt{n-m}(n/m)^{1/\alpha}+c_1\sqrt{m}n^{1/\alpha}} \right)\\ \nonumber
& \sim \Omega\left( \frac{c_0n}{c_1\sqrt{n}(n/m)^{1/\alpha}+c_1\sqrt{m}n^{1/\alpha}}\right)
\end{align}
The optimization problem in \eqref{eq:aopamdoddk} can be handled in a simple manner as follows. Since the time needed for the two phases exhibit a trade-off behavior which can be managed by changing $m$ (i.e. increasing $m$ will decrease the time of the first phase, and will increase the time of the second phase, simultaneously. ), it can be seen that, the best choice of $m$ happens when  we try to make the time needed for the two phases equal (in the scaling sense). Thus we will set
\begin{equation}\label{Eq_Enhanced_Main_Theorem_Proof_10}	
	c_1\sqrt{n}(n/m)^{1/\alpha}=c_1\sqrt{m}n^{1/\alpha}
\end{equation}
which will result in
\begin{equation}\label{Eq_Enhanced_Main_Theorem_Proof_11}	
	m=n^{\alpha / (\alpha+2)}
\end{equation}
by putting \eqref{Eq_Enhanced_Main_Theorem_Proof_11} into \eqref{Eq_Enhanced_Main_Theorem_Proof_9} (and after some simple calculations) we will have:
\begin{equation}\label{Eq_Enhanced_Main_Theorem_Proof_12}	
T_3\sim \Omega\left( n^{(\alpha^2+2\alpha-4)/(2\alpha^2+2\alpha)} \right)
\end{equation}

\end{proof}

It should be noted that, in the proof of theorems \ref{theorem:soddlwldwd} and \ref{theorem:msdsdsdsd}, we have used the technique of introducing dummy information, and hence, the results serve as lower bounds for the network throughput. As mentioned above we proved that 2-partitioning of traffic demands improves the  throughput for network with heterogeneous rate vector. We devised the optimal 2-partitioning of the rates to two parts. An interesting line of research that looks NP hard is seeking for a potentially k- partitioning of the rates to k parts such that attains true optimal improvement of the  throughput.

\section{Simulation}
\label{sec:Simulation}
We carried out extensive simulations to evaluate our proposed approaches using OMNET++ \cite{Varga1999Using} and MATLAB. Numerous validation experiments have been established. However, for the sake of specific illustration, validation results are presented for limited number of scenarios. We adopted 95 percent confidence level to make sure that, on average, the confidence interval which is calculated using $t$-student distribution and standard error contains the true values around 95 percent of the time. Moreover, Box plots are presented in the figures to illustrate detailed
behavior of the simulation experiments. The median, the 25th and 75th percentiles are
depicted. To give additional information about the spread of the results we put vertical lines above and below of each box to denote the the 9th percentile and the 91st percentile.  Outliers are plotted as individual points \cite{ross2009Introduction}.

\subsection{Throughput for network with  homogeneous rate vector}
\begin{figure*}
\begin{subfigure}[b]{0.48\textwidth}
\centering
\includegraphics[width=\linewidth]{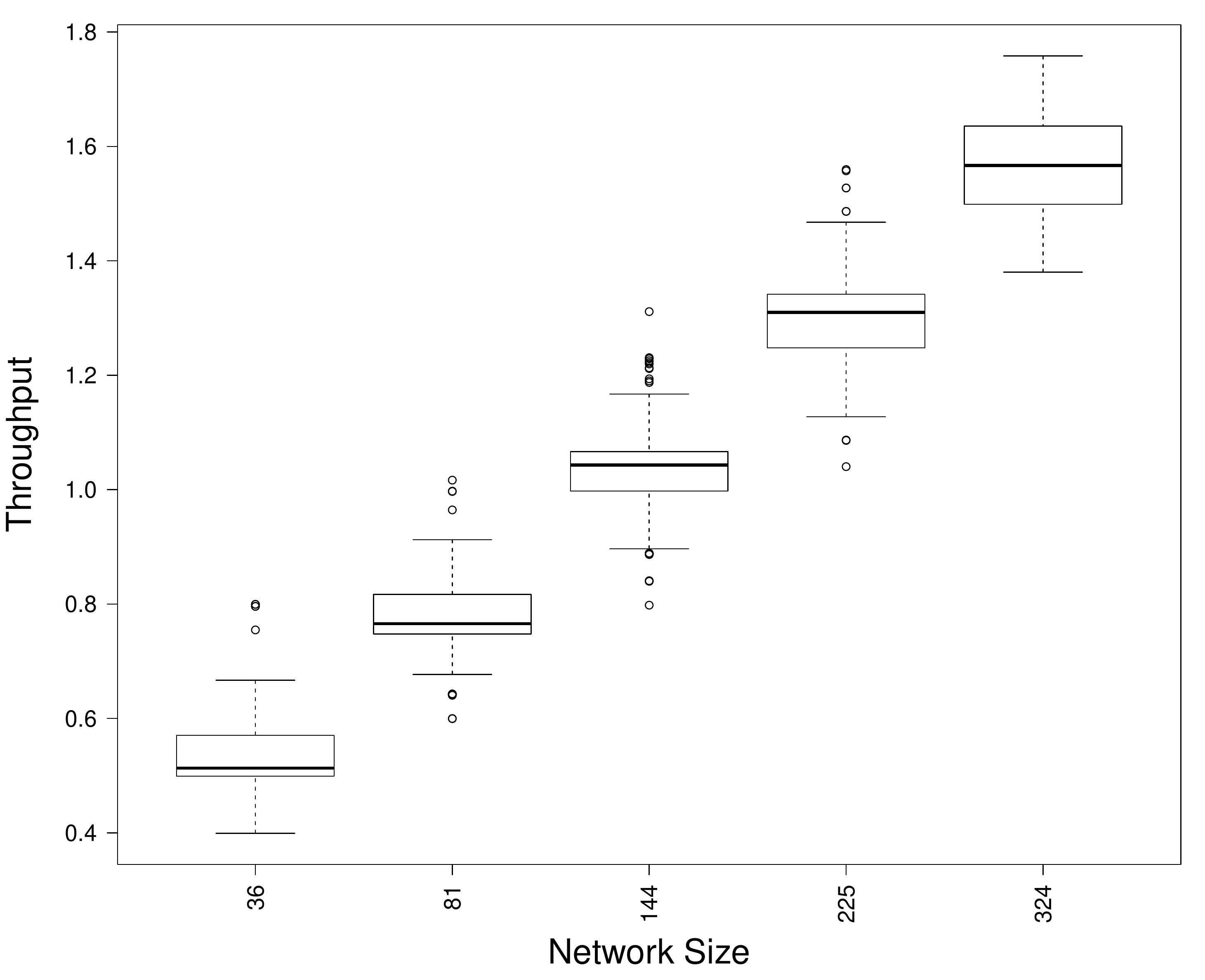}
\caption{}
\label{fig:boxplotThroughput}
\end{subfigure}
\begin{subfigure}[b]{0.48\textwidth}
\centering
\includegraphics[width=\linewidth]{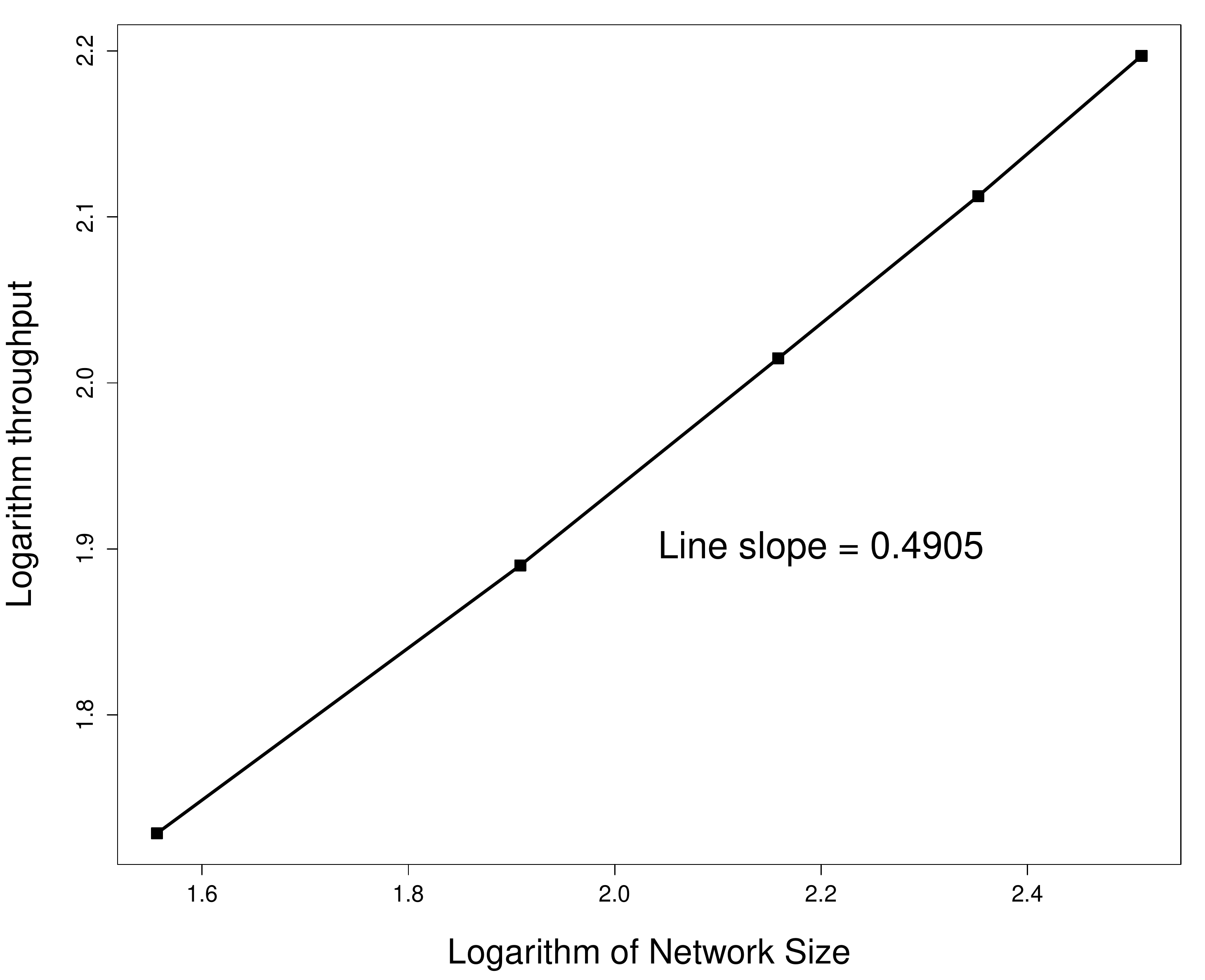}
\caption{}
\label{fig:AverageThroughput}
\end{subfigure}
\begin{subfigure}[b]{0.48\textwidth}
\centering
\includegraphics[width=\linewidth]{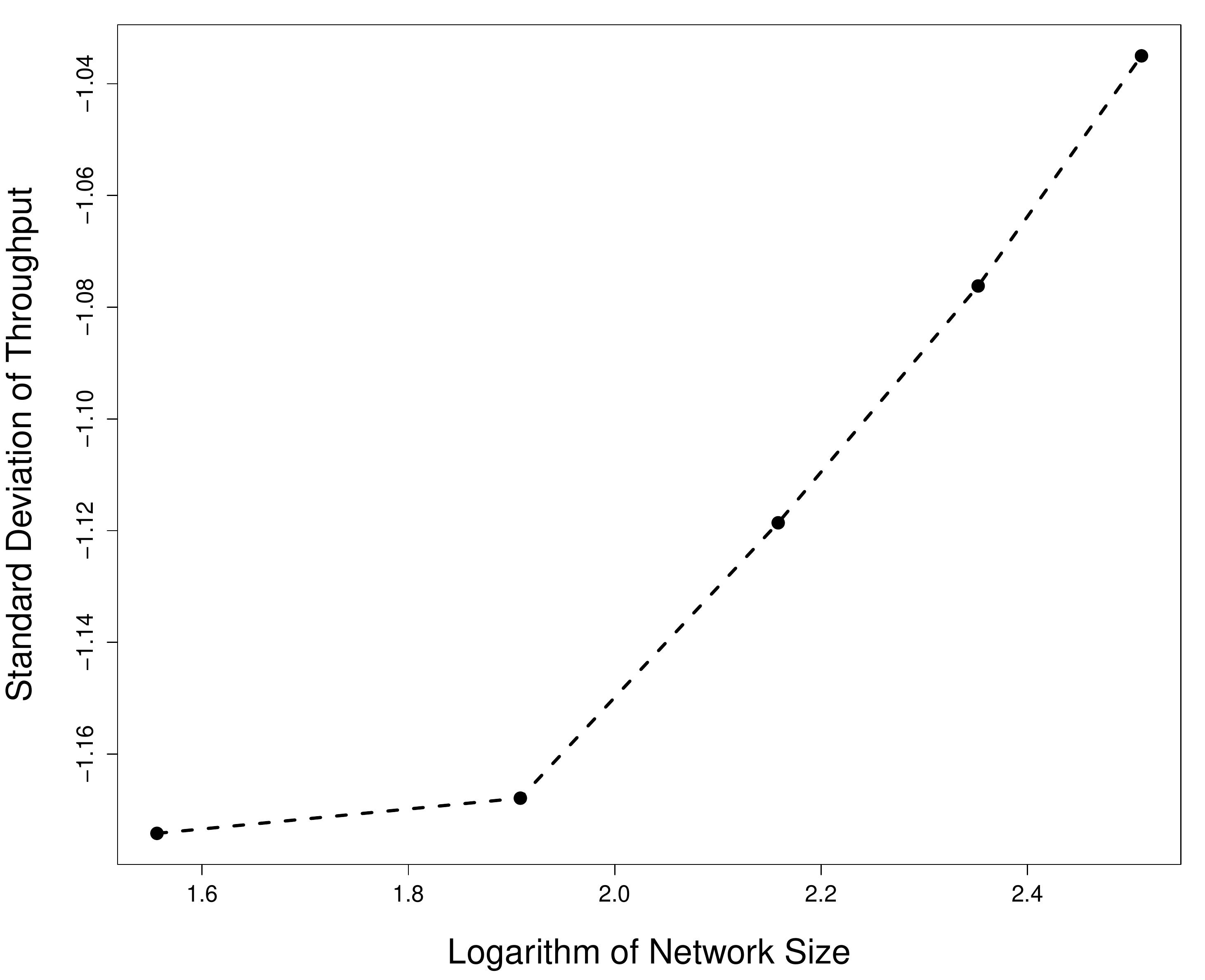}
\caption{}
\label{fig:varianceThroughput}
\end{subfigure}
\caption{Validation of all simulation experiments against the theoretical results in \cite{Gupta2000capacity} for different network sizes.  (b) Log-Log plot of throughput against network size (c) The simulation results of the standard deviation rate-homogeneous scaling law of  \cite{Gupta2000capacity}.}
\end{figure*}
Numerous validation experiments have been performed for several combinations of network sizes, and demand rates. For brevity, network sizes of $n=3\times3$ and $n=18\times18$ are presented and the nodes are arranged in a mesh with dimension $\sqrt{n}$. In accordance with \cite{Gupta2000capacity} messages are generated at each node according to a rate-homogeneous demand. In agreement with assumptions in Section \ref{subsec:Assumptions},  X-Y routing  and 9-TDMA  are employed and source-destination pairs are determined using a uniform random number generator such that they constitute uniformly distributed sessions in the network.

Each session has 100 packets to transfer. Each simulation experiment was run until all packets are delivered to their destination. Furthermore, we considered a transmission channel without packet loss. The simulation experiments were repeated 200 times, each with a new configuration of sessions that are uniformly distributed across the network.

\figurename \ref{fig:boxplotThroughput} shows the simulation results of throughput against the network size. The Box plot denotes the variability of the results in the simulation experiments. The results of aggregate throughput of 200 simulation experiments for each network size are obtained as well. These results are depicted in \figurename \ref{fig:AverageThroughput} for different network sizes using Log-Log plot. The fitting line slope for this plot is 0.4905 that matches with the slope 0.5 of Lemma \ref{Lemma:twouniformly} with a high degree of accuracy.  This result reveals that simulation experiments follow the mathematical analysis in \cite{Gupta2000capacity} with high precision confirming the scaling property of the network throughput.  Finally, the Log-Log plot of standard deviation of the aggregate throughput is shown in \figurename \ref{fig:varianceThroughput}. This result, to the best of our knowledge, has not been reported in earlier studies, revealing the scaling law for the variance of the throughput which may be useful for the potential applicability of WDCs  for streaming and multimedia applications that require stable rates.

\subsection{Throughput for network with heterogeneous rate vector}
In the next step, we consider the previous simulation scenario, with the difference that $n-1$ nodes send $100$ packets, and just one node send $100\times g(n)$ packets, and run for $200$ distinct source-destination configurations.

 \tablename \ref{table:spdpsdxwsewd} reports the fitted line slope for  aggregate throughput. According to Theorem \ref{theorem:soddlwldwd}, for $g(n)=n^{2/3}$ and $g(n)=n^{4/5}$ ($g(n)>>\sqrt{n}$),  aggregate throughput scales like $\frac{n}{g(n)}$. For $g(n)=n^{1/3}$ scenario, throughput scales like $\sqrt{n}$ . As you see in  \tablename \ref{table:spdpsdxwsewd}. the simulation results confirm  Lemma \ref{theorem:sldsdwwdazz} result. 
  \begin{table}
  \caption{Comparing the theoretical results of rate vector $\eta_{n}(1, \ldots, g(n))$ for different $g(n)$ against simulation results}
  \centering
  \renewcommand{\arraystretch}{1.5}
 \begin{tabular}{|c|c|c|c|}\hline
$g(n)$ & $n^{2/3}$ & $n^{1/3}$ & $n^{4/5}$ \\\hline
Theorem \ref{theorem:soddlwldwd} result & $n^{1/3}$ & $n^{0.5}$ & $n^{0.2}$ \\\hline
Simulation result &$n^{0.3300}$ & $n^{0.43}$ &$n^{0.1656}$ \\\hline
Exponent difference & 0.0033 & 0.07 &   0.0344 \\\hline
 \end{tabular}
 \label{table:spdpsdxwsewd}
\end{table}

In the next scenario, all nodes have different demand rate. The demand rate vector $\mathbf{r}=\eta_{n}(\lambda_1, \lambda_2 , \ldots , \lambda_n)$, is according to a heavy-tailed distribution with parameter $\alpha$.  We run each simulation until the last packet is received to its destination.
The simulation results for $\alpha =2$ and $\alpha=5$ for over the $200$ simulations, are shown in \figurename \ref{fig:boxplotThroughputR} and \ref{fig:AverageThroughputR}. In the Theorem \ref{theorem:sdowdlwdw}, we obtained a lower bound for throughput ($\Omega\left(n^{1/2-1/\alpha}\right)$). 
\begin{figure*}
\begin{subfigure}[b]{0.48\textwidth}
\centering
\includegraphics[width=\linewidth]{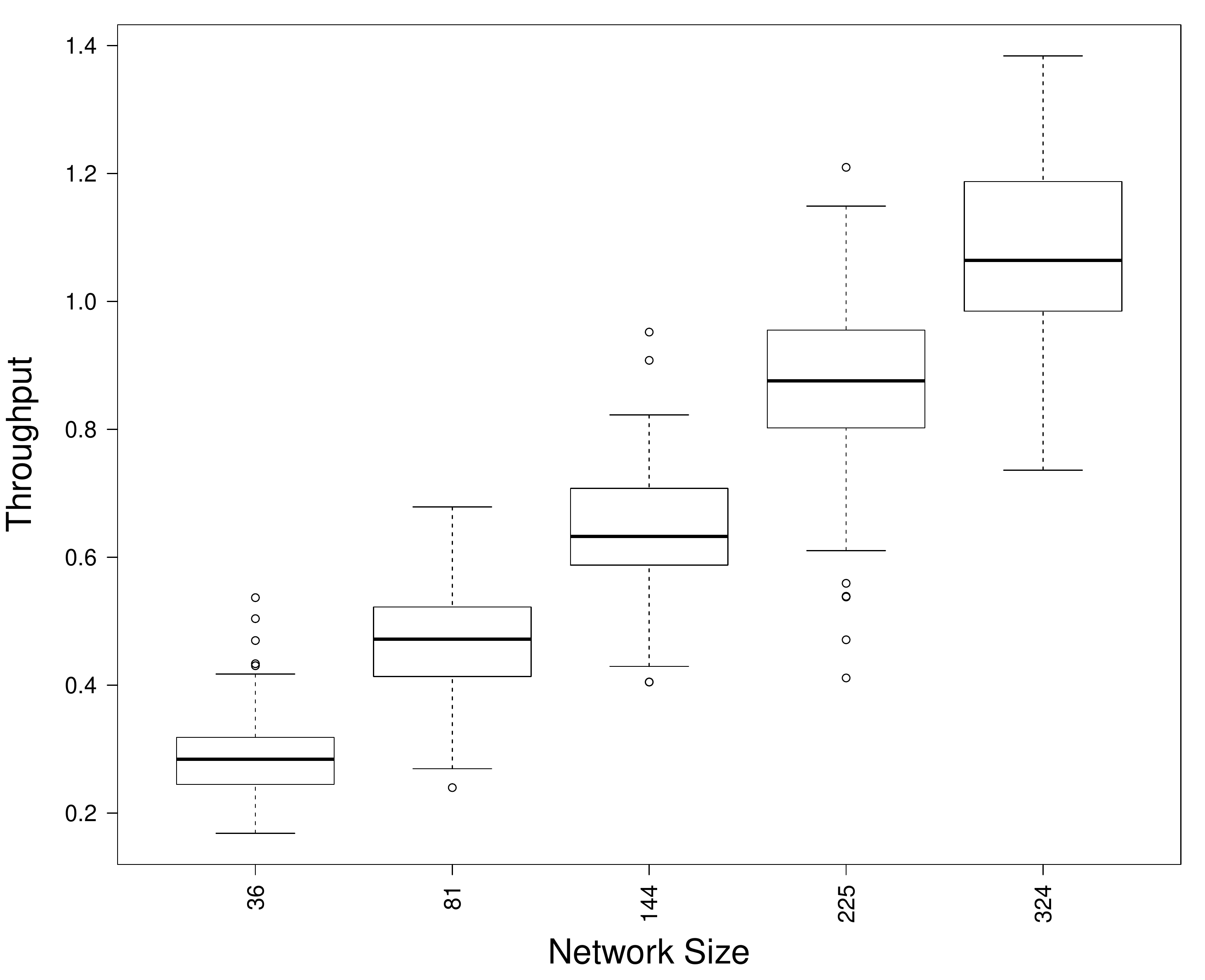}
\caption{}
\label{fig:boxplotThroughputR}
\end{subfigure}
\begin{subfigure}[b]{0.48\textwidth}
\centering
\includegraphics[width=\linewidth]{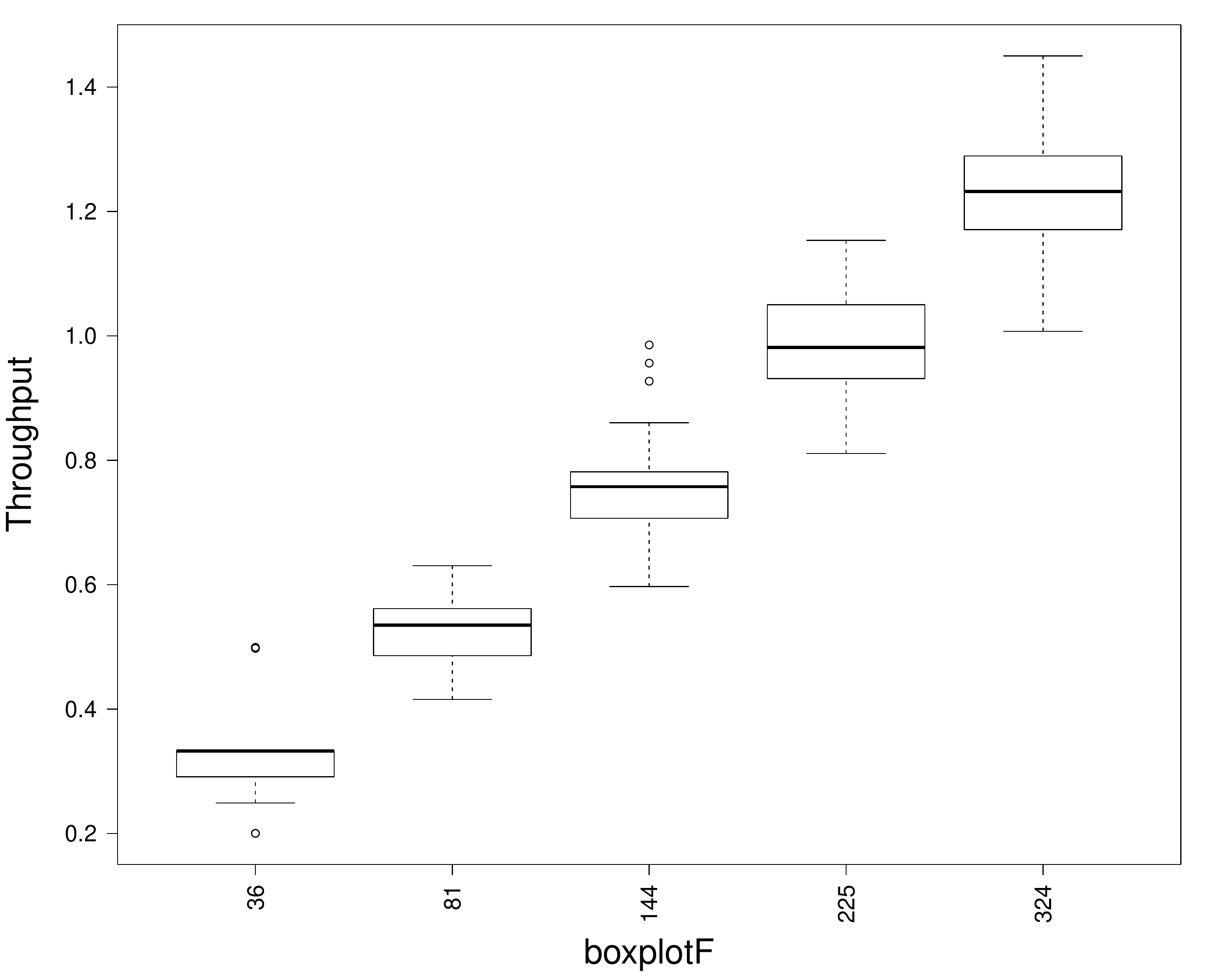}
\caption{}
\label{fig:AverageThroughputR}
\end{subfigure}\\
\begin{subfigure}[b]{0.48\textwidth}
\centering
\includegraphics[width=\linewidth]{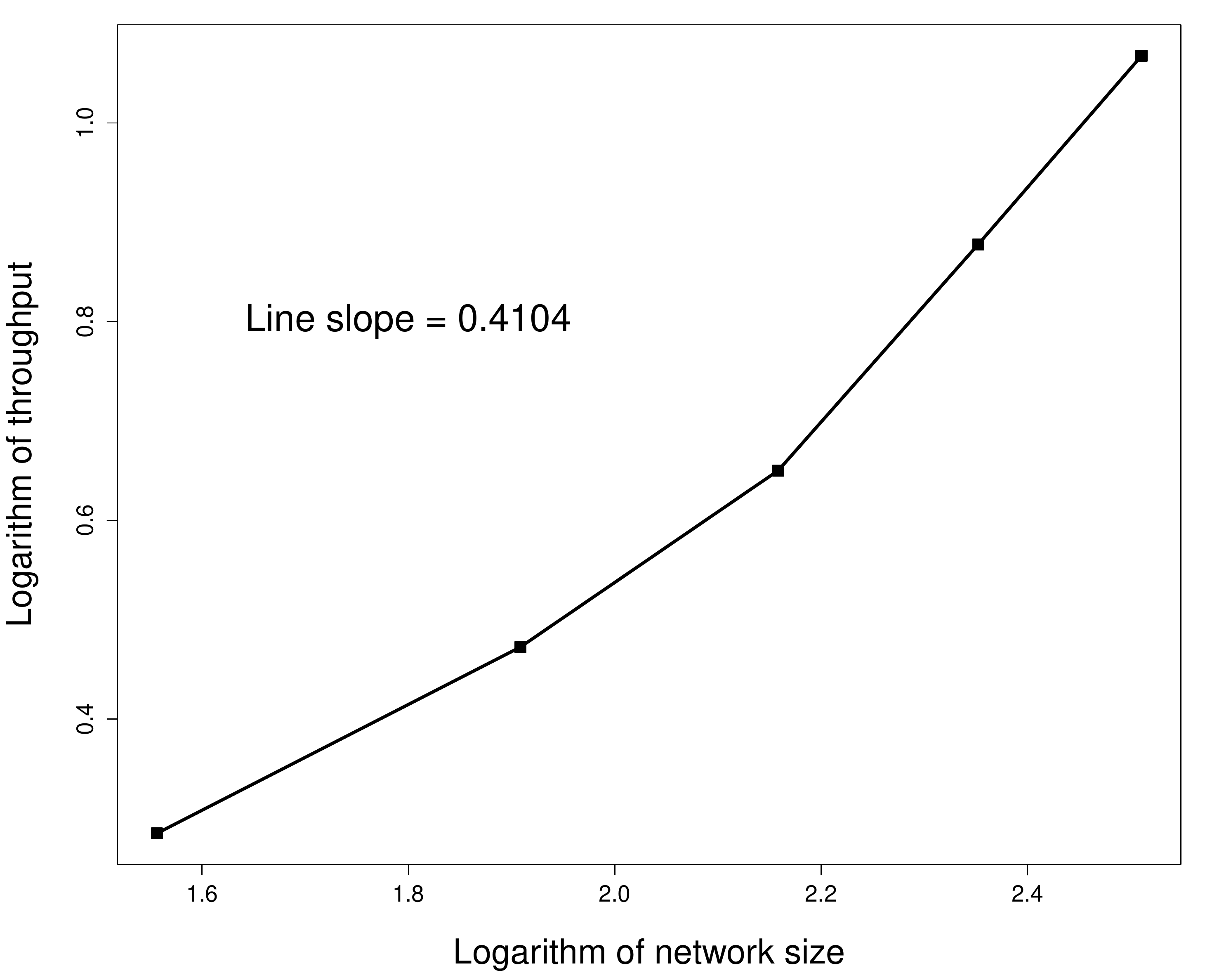}
\caption{}
\label{fig:boxplotThroughputR}
\end{subfigure}
\begin{subfigure}[b]{0.48\textwidth}
\centering
\includegraphics[width=\linewidth]{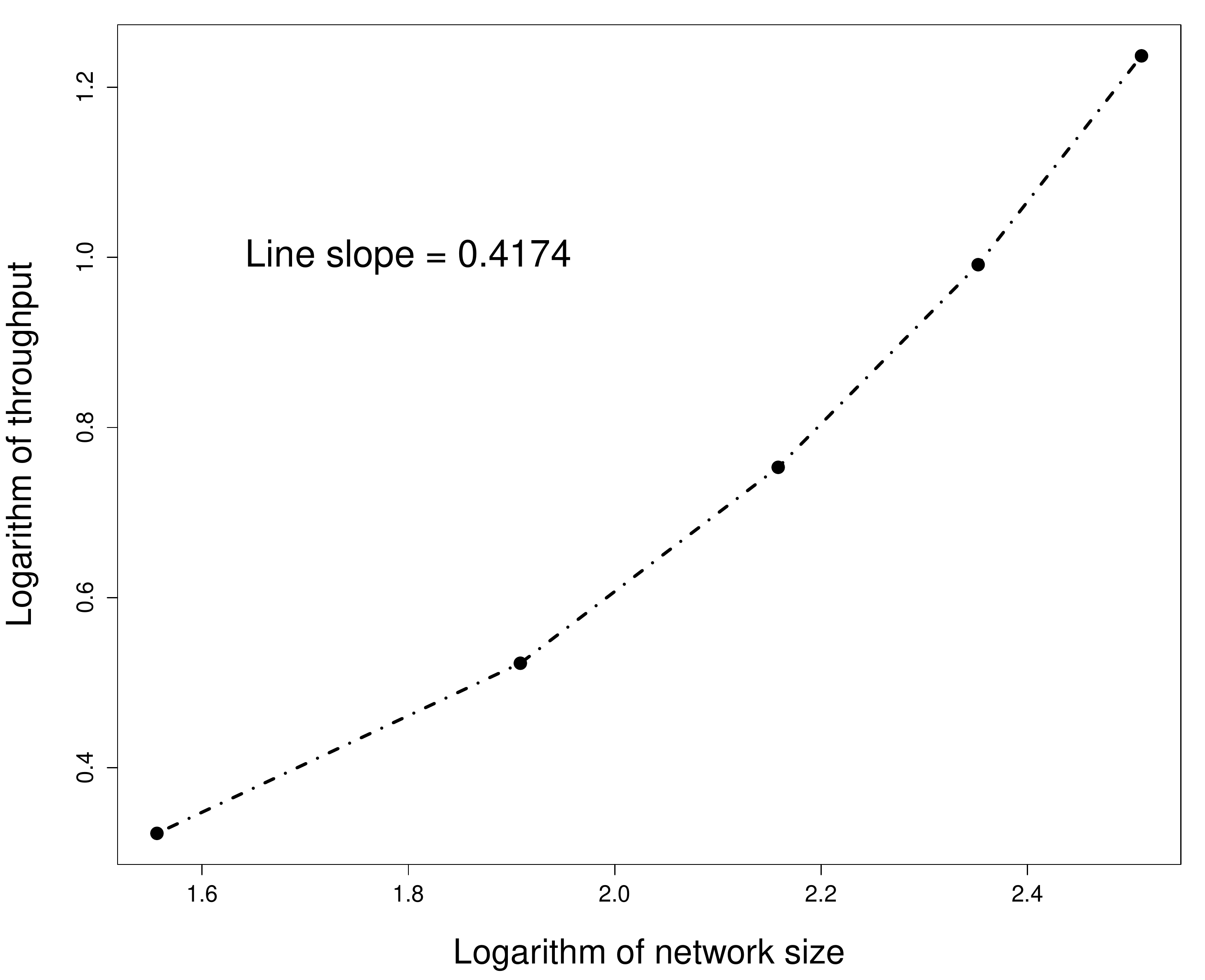}
\caption{}
\label{fig:AverageThroughputR}
\end{subfigure}
\caption{Box plot of simulation results of the throughput against network size for the rate sampled from heavy-tailed distribution with (a) $\alpha=2.5$ (b) $\alpha=5$
Log-Log plot of the simulation results of the throughput against network size for the rate sampled from heavy-tailed distribution with c) $\alpha=2.5$ (d) $\alpha=5$
}
\end{figure*}

\section{Discussion and Future work}
\label{sec:DiscussionFuture}
\begin{figure}
\centering
\includegraphics[width=0.7\linewidth]{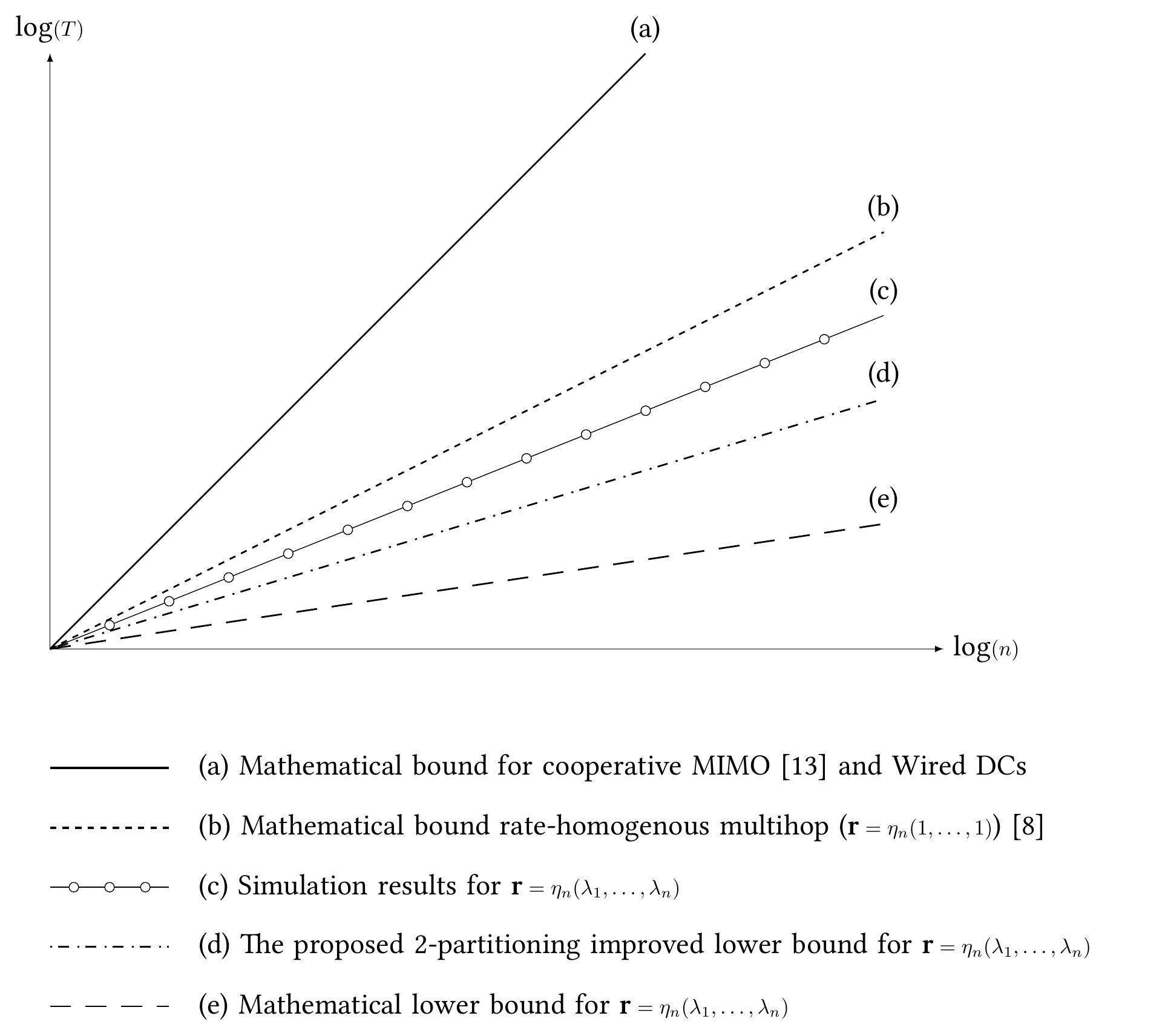}
\caption{The throughput scalability of the wireless network for different communication schemes compared to throughput scalability of wired DCs (theoretical and simulation results)  }
\label{fig:communicationScheme}
\end{figure}
Using any other MACs other than 9-TDMA cannot improve network throughput scaling as a function of the network size. Assuming a hypothetic and ideal MAC that allows every node to send its traffic in every time slot, the achieved throughput under this  ideal MAC is at most  9 times higher than 9-TDMA. This improvement of the throughput  does not depend on network size $n$, but a constant. Omitting this constant factor results in the same throughput scaling as 9-TDMA.  This is true for routing as well. Even, optimal routing improves the throughput by a constant factor, which in turn results in the same throughput scaling as X-Y routing. The above statements have been confirmed in the pioneering work of Gupta and Kumar \cite{Gupta2000capacity}. Thus we require a scale-variant solution other than the above scale-invariant improvements to make the throughput scalable with network size. This diminishes the WDCs dubiousness as a potential candidate for practical purposes. As proved in Section \ref{sec:sdjfekdfke} we noticed that 2-partitioning improves  throughput scaling. 

\figurename \ref{fig:communicationScheme} illustrates the potentials of WDCs conventional multi-hopping network with spatial reuse for different communications schemes. The figure plots logarithm throughput against logarithm network size (Log-Log plot). \figurename \ref{fig:communicationScheme}.(a) shows the linear throughput scaling of Wired DCs. \cite{Ozgur2007Hierarchical} proves that wireless cooperative MIMO under rate-homogeneous achieves this linear throughput scaling as well,  which makes WDCs comparable to novel wired DCs at least in theory.  \figurename \ref{fig:communicationScheme}.(b) illustrates that the throughput of conventional multi-hop with spatial reuse for homogeneous rate $\mathbf{r}=\eta_{n}(1, \ldots, 1)$  scales as $\sqrt{n}$. This matches with the results obtained in \cite{Gupta2000capacity}. 
\figurename \ref{fig:communicationScheme}.(e) depicts the mathematical obtained for the rate-heterogeneous vector demands  ($\mathbf{r} = \eta_{n}(\lambda_{1}, \ldots, \lambda_{n})$). It shows that the rate-heterogeneity which is a requisite in practical DCs exacerbates  the  throughput of conventional multi-hop with spatial reuse  of \figurename \ref{fig:communicationScheme}.(b). This makes the potential  applicability of wireless communication for practical DCs questionable. \figurename \ref{fig:communicationScheme}.(c) depicts simulation results for $\mathbf{r}=\eta_{n}(\lambda_1, \ldots, \lambda_n)$ and validates the mathematical analysis.

Fortunately, as depicted in \ref{fig:communicationScheme}.(d) the 2-partitioning scheme proposed in this paper sheds the light that it is possible to move forward to the bounds presented in \figurename \ref{fig:communicationScheme}.(b). The graphs for fully rate-heterogeneous rate demands 
($\mathbf{r}=\eta_{n}(\lambda_{1}, \ldots, \lambda_{n})$)
for different parameters $\alpha$, behaves as the results presented in \figurename \ref{fig:communicationScheme} and thus we do not illustrate them for brevity. Knowing that the wireless cooperative MIMO under homogeneous rate achieves  linear throughput scaling brings the hope that this solution may be a potential candidate for rate-heterogeneous traffic demands of practical DCs.  The problem of determining the throughput of the wireless cooperative MIMO and its scalability under  rate heterogeneity  is an interesting problem that is yet to be investigated. 

When everything seems satisfactory in terms of throughput scaling, the problem of latency remains still prohibitive. The results presented in Table \ref{table:prevnexthop} reveal that high throughput is achieved in a high amount of time.
We can say that although, increasing spatial reuse, rises the throughput of the network, it grows the number of hops for message delivery and thus translates to high latency. These observations suggest that knowing the insensitivity of throughput scaling  to routing and MAC protocols, efficient design of very dense WDCs require competent design of every layer of protocol stacks and topology so as to handle both high throughput and low latency in order to compete with wired DCs counterparts. 

\section{Conclusion}
\label{sec:conclusion}

We employed the theory of throughput scaling law, which is primarily proposed for  rate-homogeneous demands to investigate the potentials of wireless interconnection to deploy WDCs as a promising solution especially for small to medium scale DCs.  We obtained the standard deviation of the aggregate throughput for this regime. The results are useful for the potential applicability of WDCs for streaming and multimedia applications that require stable rates. The asymptotic throughput of rate-heterogeneous demands is obtained for two different types of heterogeneity in order to study the performance and feasibility of practical DCs that possess rate-heterogeneous traffic nature. The lower bound for the rate vector demands ($\eta_{n}(1, 1, \ldots, g(n))$) is calculated as $\Omega\left(\sqrt{n}/g(n)\right)$. The $\mathbb{E}\left\{\frac{T_1\sum_i{\lambda_i}}{n\max_i{\lambda_i}}\right\}$ lower bound has been obtained for $\mathbf{r}(n) = (r_1(n), r_2(n), \ldots, r_n(n))$ , where the rates are independently sampled from a heavy-tailed distribution. 

Both above results indicate that the throughput of WDCs under the mentioned traffic patterns does not scale with $n$ and thus are not competitive with nowadays wired DCs proposals that possess linear scalability. Thus we proposed a speculative 2-partitioning scheme so as to improve the performance of conventional multi-hopping and then obtained an improved scaling for its throughput. For $\eta_n(1, 1, \ldots , g(n))$, and when $\sqrt{n} << g(n) << n$, the aggregate throughput of our proposed scheme improves to $n/g(n)$.
For $\mathbf{r}(n) = (r_1(n), r_2(n), \ldots , r_n(n))$, where the rates are independently sampled from a heavy-tailed distribution, with parameter $\alpha$, the lower bound of the throughput is improved to $\Omega\left( n^{(\alpha^2+2\alpha-4)/(2\alpha^2+2\alpha)} \right)$. Although the results are promising, they are not comparable to nowadays wired DCs. The effects of MIMO and optimal $k$-partitioning approaches are yet to be investigated to answer the potentials of WDCs as an alternative solution of practical data centers or even interconnection network for general multicomputers.

\appendices
\label{appendix:distribution}
\section{Lemma proofs}
\subsection{Proof  lemma \autoref{Lemma:twouniformly}}
\label{append:Lemma:twouniformly}
	Consider two uniformly randomly chosen nodes on the grid as in \figurename \ref{Fig_Grid_3}. consider a disk around one of them with the radius $n^{1/2-\epsilon}$. Define $n_1$ and $n_2$ to be the number of nodes inside and outside of this disk respectively. Then the probability that the second node is outside of this disk is equal to
\begin{eqnarray}\label{Eq_UThourghput_Lemma_Proof_1}	\frac{n_2}{n}=1-\frac{n_1}{n}&=&1-\frac{\pi\left(n^{1/2-\epsilon}\right)^2 \lambda_0}{n} \\ \nonumber
&=&1-c_0n^{-2\epsilon} \rightarrow 1, \hspace{5mm} \mathrm{as} \hspace{3mm} n \rightarrow \infty
\end{eqnarray}
where $\lambda_0$ and $c_0$ are constants. Thus the second node is, with high probability, outside of this disk. This shows that the average distance of two uniformly randomly chosen nodes on the grid is greater than $n^{1/2-\epsilon}$, for any arbitrarily small $\epsilon$.
\begin{figure}
\centering
\includegraphics[width=0.40\linewidth]{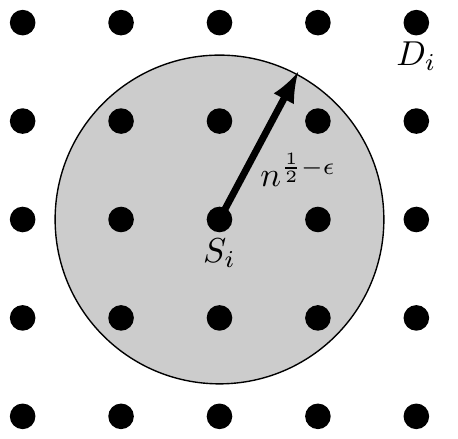}
\caption{Proof Sketch.}
\label{Fig_Grid_3}
\end{figure}

\subsection{Proof  lemma \ref{lemma:inordertoprove}}
\label{append:inordertoprove}
In order to prove this lemma we need the following \emph{extreme order statistics} lemma:
\begin{lemma}
	Consider i.i.d. random variables $\lambda_1,\dots,\lambda_n$ with the c.d.f. indicated in \eqref{Eq_Non_UThourghput_HT_pdf}. Then we will have
\begin{equation}\label{Eq_Non_UThourghput_Lemma_2_Proof_1}
	F_{\lambda}^{n}(b_nx) \rightarrow G(x)
\end{equation}
where $b_n=F^{-1}(1-1/n)$, and
\begin{equation}\label{Eq_Non_UThourghput_Lemma_2_Proof_2}
	G(x)=\exp(-x^{-\alpha}), x\geq 0
\end{equation}
and $F_{\lambda}^{n}(.)$ is the c.d.f. of $\max_i{\lambda_i}$.
\end{lemma}
\begin{lemmaproof}
For the proof refer to \cite{arnold2008first}.

\end{lemmaproof}
Then we will have
\begin{eqnarray}\label{Eq_Non_UThourghput_Th_Proof_3}
	\Pr\{Y<x\} &\triangleq& \Pr\left\{\left(\max_i{\lambda_i}/n^{1/\alpha}\right)<x\right\} \\ \nonumber
&=& \Pr\left\{\left(\max_i{\lambda_i}/b_n\right)<x\right\} \\ \nonumber
&=&\Pr\left\{\max_i{\lambda_i}<b_nx\right\} \\ \nonumber
&=& G(x)
\end{eqnarray}
Thus we will have
\begin{eqnarray}\label{Eq_Non_UThourghput_Th_Proof_4}
\mathbb{E} \left\{ \frac{1}{\left(\max_i{\lambda_i}/n^{1/\alpha}\right)} \right\}&=& \mathbb{E} \left\{ \frac{1}{Y} \right\} \\ \nonumber
&=& \alpha \int_{0}^{\infty}{\frac{\exp(-x^{-\alpha})}{x^{\alpha+2}}\mathrm{d}x} \\ \nonumber
&=& cte.
\end{eqnarray}

\subsection{Proof  lemma \autoref{Lemma:considertheoriginal}}
\label{append:considertheoriginal}
  Divide the original network into cells of $k$ nodes. Suppose we select each node to be in the final network with the probability $m/n$, independently. Then we will show that in each cell at least one node is maintained in the final network, provided that $km/n \rightarrow \infty$. The probability of this event is:
 \begin{eqnarray}\label{Eq_Enhanced_Main_Theorem_Lemma_Proof_1}
 	1-\left(1-\frac{m}{n}\right)^k&=&1-\left(1-\frac{m}{n}\right)^{(n/m)(mk/n)}\\ \nonumber
 &\rightarrow& 1-e^{-km/n} \\ \nonumber
 &\rightarrow& 1
 \end{eqnarray}
 Now, as before we have a grid network of $m$ cells and can operate it to arrive at the throughput $\sqrt{m}$.

 \subsection{Proof  lemma \autoref{Lemma:assumethatFalk1989}}
 \label{append:assumethatFalk1989}
 Consider the following lemma:
 \begin{lemma} [Falk, 1989] \label{Lemma_Falk}
 Assume that $X_1,X_2,\dots,X_n$ are i.i.d. random variables with the c.d.f. $F(x)$. Define $X_{(1)},X_{(2)}, \dots, X_{(n)}$ to be the order statistics of $X_1,X_2,\dots,X_n$. If $i \rightarrow \infty$ and $i/n \rightarrow 0$ as $n \rightarrow \infty$, then there exist sequences $a_n$ and $b_n>0$ such that
 \begin{equation}\label{Eq_Enhanced_Main_Theorem_Lemma_Proof_3}
 	\frac{X_{(n-i+1)}-a_n}{b_n} \Rightarrow \mathrm{N}(0,1),
 \end{equation}
 where $\Rightarrow$ denotes convergence in distribution, and $\mathrm{N}(0,1)$ is the Normal distribution with zero mean and unit variance. Furthermore, one choice for $a_n$ and $b_n$ is:
 \begin{eqnarray}\label{Eq_Enhanced_Main_Theorem_Lemma_Proof_4}
 	a_n=F^{-1}\left(1-\frac{i}{n}\right), \;\;\;	b_n=\frac{\sqrt{i}}{nf(a_n)}.
 \end{eqnarray}
 \end{lemma}
 \begin{lemmaproof}
 For the proof refer to \cite{arnold2008first}
 
 \end{lemmaproof}
 By applying the above lemma to the distribution specified in \eqref{Eq_Non_UThourghput_HT_pdf}, and putting $i=m$, we will have:
 \begin{eqnarray}\label{Eq_Enhanced_Main_Theorem_Lemma_Proof_5}
 	a_n &\rightarrow& \left(\frac{n}{m}\right)^{1/\alpha} \\ \nonumber
 	\frac{b_n}{a_n}&=&m^{-1/2} \rightarrow 0
 \end{eqnarray}

\bibliographystyle{ieeetr}
\bibliography{library}

\end{document}